\documentclass[copyright,creativecommons,sharealike]{eptcs}
\providecommand{\event}{QAPL 2011} 
\usepackage{breakurl}             

\usepackage{cite}
\usepackage{mparhack}
\usepackage[latin1]{inputenc}
\usepackage[T1]{fontenc}
\usepackage{amsmath,amsxtra,amssymb,stmaryrd}
\usepackage{bbm,mathtools}
\usepackage{xspace}

\usepackage{tikz}
\usetikzlibrary{decorations}
\usetikzlibrary{automata}
\usetikzlibrary{shapes.multipart}

\usepackage{complexity}
\usepackage[normalem]{ulem}
\usepackage{url}

\usepackage[amsmath,amsthm,thmmarks,thref]{ntheorem}
\theoremstyle{plain} 
\newtheorem{theorem}{Theorem}
\newtheorem{definition}[theorem]{Definition}
\newtheorem{proposition}[theorem]{Proposition}

\theorembodyfont{\upshape}

\newcommand*\Strat{\Theta}
\newcommand*\Stratblind{\tilde \Strat}
\newcommand*\stratA{\Strat_1}
\newcommand*\stratB{\Strat_2}
\renewcommand*\K{\mathbbm{K}}
\newcommand*\Real{\mathbbm{R}}
\DeclareMathOperator{\Tr}{Tr} 
\DeclareMathOperator{\Pa}{Pa} 
\DeclareMathOperator{\fPa}{fPa} 
\DeclareMathOperator{\tr}{tr} 
\DeclareMathOperator{\len}{len} 
\newcommand*\bigmid{\mathrel{\big|}}
\newcommand*\tto[1]{\xrightarrow{#1}}
\newcommand*\ie{\textit{i.e.}} 
\newcommand*\eg{\textit{e.g.}}
\newcommand*\cf{\textit{cf.}}
\newcommand*\vblind{\tilde v}
\newcommand*\last{\textup{last}} 
\newcommand{\Round}[1]{\ensuremath{\textup{Round}_{(#1)}}\xspace} 

\begin{document}


\title{Distances for Weighted Transition Systems: \\ Games and
  Properties}

\author{%
  Uli Fahrenberg 
  \institute{%
    IRISA/INRIA%
    \thanks{%
      Most of this work was conducted while this author was still at
      Aalborg University.} \\
    Rennes Cedex \\ France}
  \email{ulrich.fahrenberg@irisa.fr}
  \and
  Claus Thrane \qquad\qquad Kim G.~Larsen
  \institute{
    Department of Computer Science \\ Aalborg University \\
    Denmark}
  \email{\{crt,kgl\}@cs.aau.dk}
 }

\def\titlerunning{Distances for Weighted Transition Systems}
\def\authorrunning{Fahrenberg, Thrane \& Larsen}

\maketitle

\begin{abstract}
  We develop a general framework for reasoning about distances between
  transition systems with quantitative information.  Taking as
  starting point an arbitrary distance on system traces, we show how
  this leads to natural definitions of a linear and a branching
  distance on states of such a transition system.  We show that our
  framework generalizes and unifies a large variety of previously
  considered system distances, and we develop some general properties
  of our distances.  We also show that if the trace distance admits a
  recursive characterization, then the corresponding branching
  distance can be obtained as a least fixed point to a similar
  recursive characterization.  The central tool in our work is a
  theory of infinite path-building games with quantitative objectives.
\end{abstract}

\section{Introduction}

In verification of concurrent and reactive systems, one generally
seeks to assert properties of systems expressed in terms of \emph{sets
  of traces} (or languages) or in terms of \emph{computation trees}.
The language point of view leads to what is generally called
\emph{linear} semantics, whereas the tree point of view leads to
\emph{branching} semantics.  These semantics are the extreme points in
a spectrum containing a number of other useful notions;
see~\cite{Glabbeek01-lbs} for an overview.

As emphasized in \cite{DBLP:conf/fm/HenzingerS06}, working with applications in
complex reactive systems or in embedded systems means that classical
notions of linear and branching equivalence (or inclusion) of
processes 
often need to be extended to accommodate \emph{quantitative}
information. This can be in relation to real-time behavior, resource
usage, or can be probabilistic or stochastic information.
In such a quantitative setting, equivalences and inclusions are
replaced by symmetric or asymmetric \emph{distances} between systems.

This approach of \emph{quantitative analysis} has been taken in
numerous papers by multiple authors, both in the real-time (or
hybrid), in the probabilistic, and in general quantitative settings,
see~\cite{DBLP:conf/qest/DesharnaisLT08,conf/concur/CernyHR10,conf/csl/ChatterjeeDH08,conf/icalp/AlfaroFS04,conf/icalp/AlfaroHM03,journals/tcs/DesharnaisGJP04,FahrenbergLT10,Giacalone90,conf/formats/2005/HenzM05,journals/jlap/ThraneFL10,DBLP:conf/stoc/Kozen83,axiomat}
for a (non-exhaustive) choice of references.  Indeed, the quantitative
approach is also useful in settings \emph{without} quantitative
information in the models, \eg~in~\cite{conf/concur/CernyHR10} various
distances related to implementation correctness of discrete systems
are considered.

The above-mentioned dichotomy between languages and trees persists in
the quantitative setting, where one hence encounters both notions of
\emph{linear} and of \emph{branching} distances.  To the best of our
knowledge, the treatment of those distances, and of the relations
between them, has so far been somewhat ad hoc.  Indeed, the general
approach appears to be to introduce some particular distances which
are relevant for a particular application and then show some useful
properties; in this paper, we try to unify and generalize some of
these approaches.

The present paper is in a sense a follow-up to previous
papers~\cite{FahrenbergLT10,journals/jlap/ThraneFL10} by the same
authors.  In those papers, we introduce and investigate three
different linear and branching distances.  A paper similar in spirit
to these is~\cite{conf/icalp/AlfaroFS04}, which analyses properties of
what we later will call the \emph{point-wise} distance for weighted
Kripke structures.  The starting point for the present paper is then
the observation of similarities between the constructions for
different types of distances, which we here generalize to encompass
all of them and to construct a coherent framework.

In this paper, we take the view that in practical applications, say in
reactive systems, the \emph{system distance} which measures adherence
to the property which we want to verify, will be specific to the
concrete domain of the application. Hence in a general framework like
the one proposed here, its description must be given as an
\emph{input}. A method to obtain the actual system distances, for some
desired level of interaction, is then prescribed by the framework.

In this paper we assume that this system distance input is given as a
\emph{distance on traces}: Given two sequences of executions, one
needs only to define what it means for these sequences to be closely
related to each other.  We show that such a \emph{trace distance}
always gives rise to natural notions both of linear and of branching
distance.

To relate linear and branching distances, we introduce a general
notion of \emph{simulation game with quantitative objectives}.  The
idea of using games for linear and branching equivalences is not
new~\cite{DBLP:conf/banff/Stirling95} and has been used in a
quantitative setting
\eg~in~\cite{conf/csl/ChatterjeeDH08,DBLP:conf/qest/DesharnaisLT08},
but here we explore this idea in its full generality.

One interesting result which we can show in our general framework is
that for all interesting trace distances, the corresponding linear and
branching distance are \emph{topologically inequivalent}.  From an
application point of view this means that corresponding linear and
branching distances (essentially) always measure very different things
and that results about one of them cannot generally be transferred to
the other.  This result -- and indeed also its proof -- is a
generalization of the well-known fact that language inclusion does not
imply simulation to a quantitative setting.

We also show that for the common special case that the trace distance
has a recursive characterization, the associated branching distance
can be obtained as a least fixed point to a similar recursive
characterization.  This is again a generalization of some standard
facts about simulation, but shows that for a large class of branching
distances, characterizations as least fixed points are available.

\subsection*{Acknowledgment}

The authors acknowledge interesting and fruitful discussions on the
topic of this work with Tom Henzinger, Pavol {\v C}ern{\'y} and Arjun
Radhakrishna of IST Austria.

\section{From Trace Distances to System Distances}

Our object of study in this work are general $\K$-weighted transition
systems (to be defined below), where $\K$ is some set of weights.  For
applications, $\K$ may be further specified and admit some extra
structure, but below we just assume $\K$ to be some finite or infinite
set.

\begin{definition}
  A \emph{trace} is an infinite sequence $\big( \sigma_j\big)_{ j=
    0}^\infty$ of elements in $\K$.  The set of all such traces is
  denoted $\K^\omega$.
\end{definition}

Note that we confine our study to \emph{infinite} traces; this is
mostly for convenience, to avoid issues with finite traces of
different length.  All our results are valid when also finite traces
are allowed and the definitions changed accordingly.  We write
$\sigma_j$ for the $j$th element in a trace $\sigma$, and $\sigma^j$
for the trace obtained from $\sigma$ by deleting elements $\sigma_0$
up to $\sigma_{ j- 1}$.

\begin{definition}
  A \emph{$\K$-weighted transition system} (WTS) is a pair $A=( S, T)$
  of sets $S, T$ with $T\subseteq S\times \K\times S$.
\end{definition}

We use the familiar notation $s\tto x s'$ to indicate that $( s, x,
s')\in T$.  Note that $S$ and $T$ may indeed be infinite, also
infinite branching.  For simplicity's sake we shall follow the common
assumption that all our WTS are \emph{non-blocking}, \ie~that for any
state $s\in S$ there is a transition $s\tto x s'$ in $T$.

A \emph{path from $s_0\in S$} in a WTS $( S, T)$ is an infinite
sequence $\big( s_j\tto{ x_j} s_{ j+ 1}\big)_{j= 0}^\infty$ of
transitions in $T$.  The set of such is denoted $\Pa( s_0)$.  We will
in some places also need \emph{finite} paths, \ie~finite sequences
$\big( s_j\tto{ x_j} s_{ j+ 1}\big)_{j= 0}^n$ of transitions; the set
of finite paths from $s_0$ is denoted $\fPa( s_0)$.  For a finite path
$\pi$ as above, we let $\len( \pi)= n$ denote its length and $\last(
\pi)= s_{ n+ 1}$ its last state.  We write $\pi_j= s_j$ for the $( j+
1)$th state and $\tr( \pi)_j$ for the $( j+ 1)$th weight in a finite
or infinite path.

A path $\pi=\big( s_j\tto{ x_j} s_{ j+ 1}\big)_{ j= 0}^\infty$ gives
rise to a trace $\tr( \pi)=\big( x_j\big)_{ j= 0}^\infty$.  The set of
(infinite) \emph{traces from $s_0\in S$} is denoted $\Tr( s_0)=\big\{
\tr( \pi)\bigmid \pi\in \Pa( s_0)\big\}$.

\subsection{Interlude: Hemimetrics}
\label{se:metrics}

Before we proceed, we recall some of the notions regarding asymmetric
metrics which we will be using.  First, a \emph{hemimetric} on a set
$X$ is a function $d: X\times X\to[ 0, \infty]$ which satisfies $d( x,
x)= 0$ for all $x\in X$ and the triangle inequality $d( x, y)+ d( y,
z)\ge d( x, z)$ for all $x, y, z\in X$.

We will have reason to consider two different notions of equivalence
of hemimetrics. Two hemimetrics $d_1$, $d_2$ on $X$ are said to be
\emph{Lipschitz equivalent} if there are constants $m, M\in \Real$
such that
\begin{equation*}
  m\, d_1( x, y)\le d_2( x, y)\le M\, d_1( x, y)
\end{equation*}
for all $x, y\in X$.  Lipschitz equivalent hemimetrics are hence
dependent on each other; intuitively, a property using one hemimetric
can always be approximated using the other.

Another, weaker, notion of equivalence of hemimetrics is the
following: Two hemimetrics $d_1$, $d_2$ on $X$ are said to be
\emph{topologically equivalent} if the topologies on $X$ generated by
the open balls $B_i( x; r)=\{ y\in X\mid d_i( x, y)< r\}$, for $i= 1,
2$, $x\in X$, and $r> 0$, coincide.  Topological equivalence hence
preserves topological notions such as convergence of sequences: If a
sequence $( x_j)$ of points in $X$ converges in one hemimetric, then
it also converges in the other.

It is a standard fact that Lipschitz equivalence implies topological
equivalence.  From an application point-of-view, topological
equivalence is interesting for showing \emph{negative} results;
proving that two hemimetrics are not topologically equivalent can be
comparatively easy, and implies that intuitively, the two hemimetrics
measure very different properties.

\subsection{Examples of Trace Distances}
\label{sec:trace_distances}

The framework we are proposing in this article takes as starting point
a \emph{trace distance} defined on executions of a weighted automaton,
\ie~a hemimetric $d_T: \K^\omega\times \K^\omega\to[ 0, \infty]$.
In this section we introduce a number of different such trace
distances, to show that the framework is applicable to a variety of
interesting examples.

\paragraph{Discrete trace distances.}
\label{ex:lts}

The discrete trace distance on $\K^\omega$ is defined by $d_T( \sigma,
\tau)= 0$ if $\sigma= \tau$ and $d_T( \sigma, \tau)= \infty$
otherwise.  Hence only equality or inequality of traces is measured;
we shall see below that this distance exactly recovers the usual
Boolean framework of trace inclusion and simulation.

If $\K$ comes equipped with a preorder $\mathord\sqsubseteq\subseteq
\K\times \K$ indicating that a label $x\in \K$ may be replaced by any
$y\in \K$ with $x\sqsubseteq y$, as \eg~in~\cite{Thomsen87}, then we
may refine the above distance by instead letting $d_T( \sigma, \tau)=
0$ if $\sigma_j\sqsubseteq \tau_j$ for all $j$ and $d_T( \sigma,
\tau)= \infty$ otherwise.  We will see later that using this trace
distance, we exactly recover the \emph{extended simulation}
of~\cite{Thomsen87}; note that something similar is done
in~\cite{DBLP:journals/dke/MedeirosAW08}. 

\paragraph{Hamming distance.}

If one defines a metric $d$ on $\K$ by $d( x, y)= 0$ if $x= y$ and $d(
x, y)= 1$ otherwise, then the sum $\sum_j d( \sigma_j, \tau_j)$ for
any pair of finite traces $\sigma$, $\tau$ of equal length is
precisely the well-known Hamming distance~\cite{Hamming50}.  For
infinite traces, some technique can be used for providing finite
values for infinite sums; two such techniques are to use \emph{limit
  average} or \emph{discounting}.  We can hence define the
limit-average Hamming distance by $d_T( \sigma, \tau)= \liminf_{ j\to
  \infty} \frac1j \sum_j d( \sigma_j, \tau_j)$, and for a fixed
discounting factor $0\le \lambda< 1$, the discounted Hamming distance
by $d_T( \sigma, \tau)= \sum_j \lambda^j d( \sigma_j, \tau_j)$.

Note that this approach can easily be generalized to other
(hemi)metrics $d$ on $\K$; indeed the discrete trace distances from
above can be recovered using $d( x, y)= 0$ if $x\sqsubseteq y$ and $d(
x, y)= \infty$ otherwise.

\paragraph{Labeled weighted transition systems.}
\label{ex:wlts}

A common example of weighted
systems~\cite{DBLP:conf/formats/BouyerFLMS08,conf/concur/CernyHR10,conf/csl/ChatterjeeDH08,journals/jlap/ThraneFL10,DBLP:conf/concur/Breugel05}
has $\K= \Sigma\times \Real$ where $\Sigma$ is a discrete set of
labels.  Hence $x=( x^\ell, x^w)\in \K$ has $x^\ell\in \Sigma$ as
discrete component and $x^w\in \Real$ as real weight.  A useful trace
distance for this type of systems is the \emph{point-wise distance},
see~\cite{conf/icalp/AlfaroFS04,journals/jlap/ThraneFL10}, given by
$d_T( \sigma, \tau)= \sup_j| \sigma^w_j- \tau^w_j|$ if $\sigma^\ell_j=
\tau^\ell_j$ for all $j$ and $d_T( \sigma, \tau)= \infty$ otherwise.
This measures the biggest individual difference between $\sigma$ and
$\tau$.

Another interesting trace distance in this setting is the
\emph{accumulated distance}~\cite{journals/jlap/ThraneFL10}, where
individual differences in weights are added up.  Again one can use
limit average or discounting for infinite sums; limit-average
accumulating distance is defined by
\begin{equation*}
  d_T( \sigma, \tau)=
  \begin{cases}
    \liminf\limits_{ j\to \infty} \frac1j \sum_j| \sigma^w_j-
    \tau^w_j| &\text{if } \sigma^\ell_j= \tau^\ell_j \text{ for all }
    j \\
    \infty &\text{otherwise}
  \end{cases}
\end{equation*}
and discounted accumulating distance, for a fixed $\lambda< 1$, by
\begin{equation*}
  d_T( \sigma, \tau)=
  \begin{cases}
    \sum_j \lambda^j| \sigma^w_j- \tau^w_j| &\text{if }
    \sigma^\ell_j= \tau^\ell_j \text{ for all } j \\
    \infty &\text{otherwise}
  \end{cases}
\end{equation*}
This is indeed a generalization of the Hamming distance above, setting
$d\big(( x, w),( y, v)\big)=| w- v|$ if $x= y$ and $d\big(( x, w),( y,
v)\big)= \infty$ otherwise.

Also of interest is the \emph{maximum-lead distance}
from~\cite{conf/formats/2005/HenzM05}, where the individual weights
are added up and one is concerned with the maximal difference between
the accumulated weights.  The definition is
\begin{equation*}
  d_T( \sigma, \tau)=
    \begin{cases}
      \sup_j \bigl| \sum_{i= 0}^j \sigma^w_i - \sum_{i= 0}^j
      \tau^w_i\bigr| &\text{if }
      \sigma^\ell_j= \tau^\ell_j \text{ for all } j \\
      \infty &\text{otherwise}
    \end{cases}  
\end{equation*}

\subsection{Simulation Games}
\label{se:wsg}

In this central section we introduce the game which we will use to
define both the linear and the branching distances.  We shall use some
standard terminology and constructions from game theory here; for a
good introduction to the subject see \eg~\cite{Ferguson}.

Let $A=( S, T)$ be a weighted transition system with $s, t\in
S$ and $d_T: \K^\omega\times \K^\omega\to[ 0, \infty]$ a trace
distance. Using $A$ as a game graph, the \emph{simulation game} played
on $A$ from $( s, t)$ is an infinite turn-based two-player game,
where we denote the strategy space of Player~$i$ by $\Strat_i$ and the
utility function of Player~1 by $u: \Strat_1 \times \Strat_2 \to[ 0,
\infty]$.  As usual $u$ determines the pay-off of Player~1; we will
not use pay-offs for Player~2 here.

The game moves along transitions in $A$ while building a pair of paths
extending from $s$ and $t$, according to the strategies of the
players.  In the terminology
of~\cite{DBLP:journals/tcs/Bodlaender93,DBLP:journals/tcs/FraenkelS93}
we are playing a \emph{partisan path-forming game}.

A \emph{configuration} of the game is a pair of finite paths $( \pi_1,
\pi_2)\in \fPa( s)\times \fPa( t)$ (\ie~the history) which are
consecutively updated by the players 
as the game progresses.
%
The players must play according to a strategy of the following types: 
\begin{itemize}
\item $\Strat_1= T^{ \fPa( s)\times \fPa( t)}$, the set of
  mappings from pairs of finite paths to transitions, with the
  additional requirement that for all $\theta_1\in \Strat_1$ and $(
  \pi_1, \pi_2)\in \fPa( s)\times \fPa( t)$, $\theta_1( \pi_1,
  \pi_2)=( \last( \pi_1), x, s')$ for some $x\in \K$, $s'\in S$.  This
  is the set of Player-1 strategies which observe the complete
  configuration.
\item Similarly, $\Strat_2= T^{ \fPa( s)\times \fPa( t)}$ with the
  additional requirement that for all $\theta_2\in \Strat_2$ and $(
  \pi_1, \pi_2)\in \fPa( s)\times \fPa( t)$, $\theta_2( \pi_1,
  \pi_2)=( \last( \pi_2), y, t')$ for some $y\in \K$, $t'\in S$.
\item $\Stratblind_1= T^{ \fPa( s)}$, the set of \emph{blind}
  Player-1 strategies which cannot observe the moves of Player~2.
  (Blind Player-2 strategies can be defined similarly, but we will not
  need those here.)  It is convenient to identify $\Stratblind_1$ with
  the subset of $\Strat_1$ of all strategies $\theta_1$ which satisfy
  $\theta_1( \pi_1, \pi_2)= \theta_1( \pi_1, \pi_2')$ for all
  $\pi_1\in \fPa( s)$, $\pi_2, \pi_2'\in \fPa( t)$.
\item In the proof of Proposition~\ref{prop:hemi} we will also need
  Player-2 strategies with additional memory.  Such a strategy is a
  mapping $\fPa( s)\times \fPa(t)\times M\to T\times M$, where $M$
  is a set used as memory.
\end{itemize}

Given a game with configuration $( \pi_1, \pi_2)$, a \emph{round} is
played, according to a \emph{strategy profile} (\ie~a pair of
strategies) $( \theta_1, \theta_2)\in \Strat_1\times \Strat_2$, by
first updating $\pi_1$ according to $\theta_1$ and then updating the
resulting configuration according to $\theta_2$.  Hence we define
\begin{equation*}
  \Round{ \theta_1, \theta_2}( \pi_1, \pi_2)=\big( \pi_1\cdot \theta_1(
  \pi_1, \pi_2), \pi_2\cdot \theta_2( \pi_1\cdot \theta_1( \pi_1, \pi_2),
  \pi_2)\big)
\end{equation*}
where $\cdot$ denotes sequence concatenation.

A strategy profile $( \theta_1, \theta_2)\in \Strat_1\times \Strat_2$
inductively determines an infinite sequence $\smash{\big(( \pi_1^j,
  \pi_2^j)\big)_{ j= 0}^\infty}$ of configurations given by $(
\pi_1^0, \pi_2^0)=( s, t)$ and $( \pi^j_1, \pi^j_2)= \Round{
  \theta_1, \theta_2}( \pi^{j-1}_1, \pi^{j-1}_2)$ for $j\ge 1$.  The
paths in this sequence satisfy $\pi_i^j\sqsubseteq \pi_i^{ j+ 1}$,
where $\sqsubseteq$ denotes prefix ordering, hence the limits $\pi_1(
\theta_1, \theta_2)= \lim_{ j\to \infty} \pi_1^j\in \Pa( s)$,
$\pi_2( \theta_1, \theta_2)= \lim_{ j\to \infty} \pi_2^j\in \Pa( t)$
exist (as infinite paths).
We define the utility function $u$ as
\begin{equation*}
  u( \theta_1, \theta_2)= d_T\big( \tr( \pi_1( \theta_1, \theta_2)), \tr(
  \pi_2( \theta_1, \theta_2))\big)
\end{equation*}
This determines the pay-off to Player~1 when the game is played
according to strategies $\theta_1$, $\theta_2$.  Note again that the
utility function for Player~2 is left undefined; especially we make no
claim as to the game being zero-sum.

The \emph{value} of game on $A$ from $( s, t)$ is defined to be
the optimal Player-1 pay-off
\begin{equation*}
  v( s, t)= \adjustlimits \sup_{ \theta_1 \in \Strat_1} \inf_{
    \theta_2 \in \Strat_2} u( \theta_1, \theta_2)
\end{equation*}
Observe that the game is \emph{asymmetric}; in general, $v( s,
t)\ne v( t, s)$.  

A strategy $\hat \theta_1$ for Player~1 is said to be \emph{optimal}
if it realizes the supremum above, \ie~whenever $\inf_{ \theta_2\in
  \Theta_2} u( \hat \theta_1, \theta_2)= \sup_{ \theta_1 \in \Strat_1}
\inf_{ \theta_2 \in \Strat_2} u( \theta_1, \theta_2)$. The strategy is
called \emph{$\epsilon$-optimal} for some $\epsilon> 0$ provided that
$\inf_{ \theta_2\in \Theta_2} u( \hat \theta_1, \theta_2)\ge \sup_{
  \theta_1 \in \Strat_1} \inf_{ \theta_2 \in \Strat_2} u( \theta_1,
\theta_2)- \epsilon$.  Note that $\epsilon$-optimal strategies always
exist for any $\epsilon> 0$, whereas optimal strategies may not.

We also recall that the game is said to be \emph{determined} if the
sup and inf above can be interchanged, \ie~if $v( s, t)= \inf_{
  \theta_2 \in \Strat_2} \sup_{ \theta_1 \in \Strat_1} u( \theta_1,
\theta_2)$.  Intuitively, the game is determined if there also exist
$\epsilon$-optimal Player-2 strategies for any $\epsilon> 0$ which
realize the value of the game (up to $\epsilon$) independent of the
strategy Player~1 might choose.

The \emph{1-blind value} of the game is defined to be
\begin{equation*}
  \vblind( s, t)= \adjustlimits \sup_{ \theta_1 \in \Stratblind_1}
  \inf_{ \theta_2 \in \Strat_2} u( \theta_1, \theta_2)
\end{equation*}

\subsection{Example: Discrete Trace Distance}
\label{se:discrete}

It may be instructive to apply the above simulation game in the
context of the discrete trace distance $d_T( \sigma, \tau)= 0$ if
$\sigma= \tau$, $d_T( \sigma, \tau)= \infty$ otherwise, from
Section~\ref{sec:trace_distances}.  In this case, the game has value
$v(s,t) = 0$ if and only if, for every $\theta_1\in\Strat_1$ there
exist a $\theta_2\in \Strat_2$ which in each round $i \ge 0$ of the
game, in configuration ($\pi^i_1,\pi^i_2$), maps
$\theta_2(\pi^i_1,\pi^i_2) = (\last(\pi^i_2),x,t_{i+1})$ whenever
$\theta_1(\pi^i_1,\pi^i_2) = (\last(\pi^i_1),x,s_{i+1})$.  Otherwise,
$v( s, t)= \infty$.

Hence we have $v( s, t)= 0$ if $t$ simulates $s$ in the sense
of~\cite{milner89}, and $v( s, t)= \infty$ otherwise.  In other
words, for discrete trace distance the game reduces to the standard
simulation game of~\cite{DBLP:conf/banff/Stirling95}.

Likewise, the blind value $\vblind(s,t) = 0$ if and only if every
$\theta_1 \in \Stratblind_1$ and corresponding path $\pi_1$ has a
match $\theta_2\in \Strat_2$ where configuration ($\pi^i_1,\pi^i_2$),
of round $i \geq 0$ facilitates $\theta_2(\pi^i_1,\pi^i_2) =
(\last(\pi^i_2),x,t_{i+1})$ whenever $\theta_1(\pi^i_1,\pi^i_2) =
(\last(\pi^i_1),x,s_{i+1})$.  Hence $\vblind( s, t)= 0$ if $\Tr(
s)\subseteq \Tr( t)$; we recover standard trace inclusion.

\subsection{Linear and Branching Distance}

We can now use the game introduced in Section~\ref{se:wsg} to define
linear and branching distance:

\begin{definition}
  \label{de:branch}
  Let $A=( S, T)$ be a WTS and $s, t\in S$.
  \begin{itemize}
  \item The \emph{linear distance} from $s$ to $t$ is the 1-blind
    value $d_L( s, t)= \vblind( s, t)$.
  \item The \emph{branching distance} from $s$ to $t$ is the value
    $d_B( s, t)= v( s, t)$.
  \end{itemize}
\end{definition}

We proceed to show that the distances so defined are hemimetrics on
$S$, \cf~the proof of Theorem~1 in~\cite{conf/concur/CernyHR10}.
For linear distance, this also follows from
Theorem~\ref{thm:dl-formula} below, and we only include the proof for
reasons of exposition.  For branching distance, we have to assume in
the proof below that the simulation game is \emph{determined};
currently we do not know whether this assumption can be lifted.

\begin{proposition}
  \label{prop:hemi}
  Linear distance $d_L$ is a hemimetric on $S$, and if the simulation
  game is determined, so is $d_B$.
\end{proposition}

\begin{proof}
  Non-negativity of $d_L$ and $d_B$ follow directly from the
  non-negativity of $d_T$.  To prove that $d_L(s,s) = d_B(s,s) = 0$
  for all $s \in S$, given any strategy $\theta_1 \in \stratA$, we
  construct a Player-2 strategy $\theta_2 \in \stratB$ that mimics
  $\theta_1$ and attains the game value of $0$, as follows:
  \begin{equation*}
    \theta_2(\pi',\pi) =
    \begin{cases}
      \theta_1(\pi,\pi) &\text{if } \pi' = \pi \cdot \theta_1(\pi,\pi)
      \\
      (\last(\pi),y,s') &\text{for some } (\last(\pi),y,s') \in T
      \text{ otherwise}
    \end{cases}
  \end{equation*}
  It can be seen easily that the paths constructed by both players are
  the same, \ie~$u(\pi_1,\pi_2) = d_T(\tau,\tau)$ for some $\tau \in
  \Pa(s_0)$. Therefore, $u(\pi_1,\pi_2) = 0$ as $d_T$ is a hemimetric,
  whence $d_L(s,s) = d_B(s,s) = 0$.

  We are left with showing that $d_L$ and $d_B$ obey the triangle
  inequality.  For linear distance, let $s_1, s_2, s_3\in S$ and write
  $\Strat_k^{ i, j}$ ($\Stratblind_k^{ i, j}$) for the set of (blind)
  Player-$k$ strategies in the simulation game computing $d_L( s_i,
  s_j)$, for $i, j\in\{ 1, 2, 3\}$ and $k\in\{ 1, 2\}$.  Let
  $\epsilon> 0$.
  It might be beneficial to look at Figure~\ref{fig:strategychase} to
  see the ``chase of strategies'' we will be conducting.

  \begin{figure}[tp]
    \centering
    \begin{tikzpicture}[shorten >=1pt, auto, initial text=, scale=.2]
      \node (s1) {\normalsize $s_1$};
      \node (pi113) [above of=s1] {$\theta_1^{ 1, 3}$};
      \node (pi112) [below left of=s1] {$\theta_1^{ 1, 2}$};
      \node (s2) [below right of=s1, xshift=3em, yshift=-3em]
      {\normalsize $s_2$};
      \node (pi212) [below left of=s2] {$\theta_2^{ 1, 2}$};
      \node (pi123) [below right of=s2] {$\theta_1^{ 2, 3}$};
      \node (s3) [above right of=s2, xshift=3em, yshift=3em]
      {\normalsize $s_3$};
      \node (pi213) [above of=s3] {$\theta_2^{ 1, 3}$};
      \node (pi223) [below right of=s3] {$\theta_2^{ 2, 3}$};
      \draw (s1) -- 
      (s2);
      \draw (s1) -- 
      (s3);
      \draw (s2) -- 
      (s3);
      \draw [->] (pi113) to [out=210, in=90] (pi112);
      \draw [->] (pi112) to [out=270, in=170] (pi212);
      \draw [->] (pi212) to [out=350, in=190] (pi123);
      \draw [->] (pi123) to [out=10, in=270] (pi223);
      \draw [->] (pi223) to [out=90, in=330] (pi213);
    \end{tikzpicture}
    \caption{%
      \label{fig:strategychase} Construction of strategies in the
      proof of Proposition~\ref{prop:hemi}}
  \end{figure}
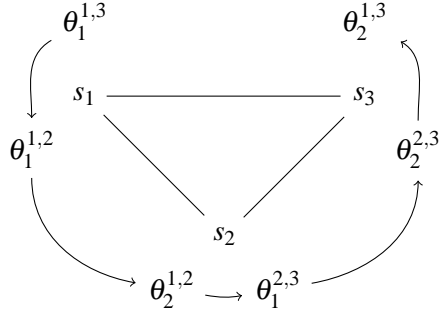

  Choose a blind Player-1 strategy $\theta_1^{ 1, 3}\in
  \Stratblind_1^{ 1, 3}$.  Blind strategies correspond to choosing a
  path, so let $\pi_1\in \Pa( s_1)$ be the path chosen by $\theta_1^{
    1, 3}$. This path in turn corresponds to a blind Player-1 strategy
  $\theta_1^{ 1, 2}\in \Stratblind_1^{ 1, 2}$.

  Let $\theta_2^{ 1, 2}\in \Strat_2^{ 1, 2}$ be a Player-2 strategy
  for which $u(\theta_1^{ 1, 2},\theta_2^{1, 2}) < d_L(s_1,s_2) +
  \frac{\epsilon}{2}$.  Write $\pi_2 \in \Pa(s_2)$ for the path
  constructed by the strategy profile $(\theta_1^{ 1, 2},\theta_2^{ 1,
    2})$, and let $\theta_1^{ 2, 3}\in \Stratblind_1^{ 2, 3}$ be a
  blind Player-1 strategy which constructs $\pi_2$.

  Let $\theta_2^{ 2, 3}\in \Strat_2^{ 2, 3}$ ensure $u(\theta_1^{ 2,
    3},\theta_2^{2, 3}) < d_L(s_2,s_3) + \frac{\epsilon}{2}$.  Write
  $\pi_3 \in \Pa(s_3)$ for the path constructed by the strategy
  profile $(\theta_1^{ 2, 3},\theta_2^{ 2, 3})$, and let $\theta_2^{
    1, 3}\in \Strat_2^{ 1, 3}$ be a strategy which constructs $\pi_3$.
  For the strategy profile $(\theta_1^{ 1, 3}, \theta_2^{ 1, 3})$ in
  $G_{ 1, 3}$, the paths constructed are $\pi_1\in \Pa( s_1)$ and
  $\pi_3\in \Pa( s_3)$.  Hence we have
  \begin{equation}
    \label{eq:triangle}
    \begin{split}
      \smash[b]{\inf_{\theta_2 \in \Strat_2^{ 1, 3}}} u(\theta_1^{ 1,
        3},\theta_2)\le
      u(\theta_1^{ 1, 3},\theta_2^{ 1, 3}) &= d_T(\pi_1,\pi_3) \\
      &\le d_T(\pi_1,\pi_2) + d_T(\pi_2,\pi_3) \\
      &= u(\theta_1^{ 1, 2}, \theta_2^{ 1, 2}) + u(\theta_1^{ 2,
        3},\theta_2^{ 2, 3}) \\ 
      &\le d_L(s_1,s_2) + d_L(s_2,s_3) + \epsilon \\
    \end{split}
  \end{equation}

  As $\theta_1^{ 1, 3}\in \Stratblind_1^{ 1, 3}$ was chosen arbitrarily,
  we have
  \begin{equation*}
    \adjustlimits \sup_{ \theta_1\in \Stratblind_1^{1, 3}} \inf_{
      \theta_2\in \Strat_2^{ 1, 3}} u( \theta_1, \theta_2)\le d_L(s_1,s_2) +
    d_L(s_2,s_3) + \epsilon
  \end{equation*}
  and as also $\epsilon$ was chosen arbitrarily, $d_L(s_1,s_3)\le
  d_L(s_1,s_2) + d_L(s_2,s_3)$.

  \smallskip%
  For branching distance, we cannot construct the paths in a one-shot
  manner as above, as the transitions chosen by Player~1 may depend on
  the history of the play.  Let again $\epsilon> 0$; assuming that the
  simulation game is determined, we can choose Player-2 strategies
  $\theta_2^{ 1, 2}\in \Theta_2^{ 1, 2}$, $\theta_2^{ 2, 3}\in
  \Theta_2^{ 2, 3}$ for which $\sup_{ \theta_1\in \Theta_1^{ 1, 2}} u(
  \theta_1, \theta_2^{ 1, 2})< d_B( s_1, s_2)+ \frac\epsilon2$ and
  $\sup_{ \theta_1\in \Theta_1^{ 2, 3}} u( \theta_1, \theta_2^{ 2,
    3})< d_B( s_2, s_3)+ \frac\epsilon2$.  Intuitively, we will use
  these strategies to allow Player~2 to find replying moves to
  Player-1 moves in the game computing $d_B( s_1, s_2)$ by using the
  reply given by $\theta_2^{ 2, 3}$ to the reply given by $\theta_2^{
    1, 2}$.  Hence we still follow the proof strategy depicted in
  Figure~\ref{fig:strategychase}, but now for individual moves.

  The strategy $\theta_2^{ 1, 3}$ uses a finite path $m= \pi_2\in
  \fPa( s_2)$ as memory and is defined by
  \begin{equation*}
    \theta_2^{ 1, 3}( \pi_1, \pi_3)( \pi_2)= \theta_2^{ 2, 3}\big(
    \pi_2\cdot \theta_2^{ 1, 2}( \pi_1, \pi_2), \pi_3\big)
  \end{equation*}
  with memory update $m( \pi_1, \pi_3)( \pi_2)= \pi_2\cdot \theta_2^{
    1, 2}( \pi_1, \pi_2)$.  The initial memory for $\theta_2^{ 1, 3}$
  is set to be the empty path, hence as the game progresses, a path
  $\pi_2\in \Pa( s_2)$ is constructed.  

  Now choose some $\theta_1^{ 1, 3}\in \Theta_1^{ 1, 3}$, and let
  $\pi_1\in \Pa( s_1)$ and $\pi_3\in \Pa( s_3)$ be the paths
  constructed by the strategy profile $( \theta_1^{ 1, 3}, \theta_2^{
    1, 3})$.  If $\pi_2\in \Pa( s_2)$ is the corresponding memory
  path, then the pair $( \pi_1, \pi_2)$ is constructed by the strategy
  profile $( \theta_1^{ 1, 3}, \theta_2^{ 1, 2})$ and the pair $(
  \pi_2, \pi_3)$ by the profile $( \theta_2^{ 1, 2}, \theta_2^{ 2,
    3})$.  Hence we can use the exact same reasoning as
  in~\eqref{eq:triangle} to conclude that
  \begin{equation*}
    \inf_{\theta_2 \in \Strat_2^{ 1, 3}} u( \theta_1^{ 1,
      3}, \theta_2)\le d_B( s_1, s_2)+ d_B( s_2, s_3)+ \epsilon
  \end{equation*}
  and hence $d_B(s_1,s_3)\le d_B(s_1,s_2) + d_B(s_2,s_3)$.
\end{proof}

\subsection{Properties}

The following general result confirms that, regardless of the trace
distance chosen, the linear distance is always bounded above by the
branching distance.  In the context of the \emph{discrete} trace
distance from Section~\ref{sec:trace_distances}, this specializes to
the well-known fact that simulation implies language inclusion.

\begin{theorem}
  \label{thm:bound}
  For any $s, t\in S$, we have $d_L( s, t)\le d_B( s, t)$.
\end{theorem}

\begin{proof}
  Any Player-1 strategy in $\Stratblind_1$ is also in $\Strat_1$, hence
  \begin{equation*}
    \adjustlimits \sup_{ \theta_1 \in \Stratblind_1} \inf_{ \theta_2
      \in \Strat_2} u( \theta_1, \theta_2)\le%
    \adjustlimits \sup_{ \theta_1\in \Strat_1} \inf_{ \theta_2\in
      \Strat_2} u( \theta_1, \theta_2)
  \end{equation*}
\end{proof}

The game definition of linear distance yields the following explicit
formula.  Note the resemblance of this to the well-known Hausdorff
construction for lifting a metric on a set to its set of subsets.

\begin{theorem}
  \label{thm:dl-formula}
  For all $s, t\in S$ we have
  \begin{equation*}
    d_L( s, t)= \adjustlimits \sup_{ \sigma\in \Tr( s)} \inf_{ \tau\in
      \Tr( t)} d_T( \sigma, \tau)
  \end{equation*}
\end{theorem}

\begin{proof}
  The definition of $\vblind( s, t)$ immediately entails the fact that
  for any $\pi_1 \in \Pa(s)$ there exists $\pi_2\in\Pa(t)$ such that
  $d_T(\tr(\pi_1),\tr(\pi_2)) \le \vblind( s, t)$. It remains to show
  that $\vblind(s,t)\le \sup_{ \sigma\in \Tr( s)} \inf_{ \tau\in \Tr(
    t)} d_T( \sigma, \tau)$.
  By blindness, any $\theta_1\in\Stratblind_1$ produces a unique path
  $\pi_1 \in \Pa(s)$ independent of the opponent strategy
  $\theta_2\in\Strat_2$.  Hence we need only consider strategies
  $\theta_2$ which define a single path $\pi_2$ from $t$, and the
  result follows.
\end{proof}

We finish this section by exposing two properties regarding
\emph{equivalence} of the introduced hemimetrics.  Transferring
(in)equivalence of distances from one setting to another is an
important subject, and we hope to show other results of the below
kind, especially relating trace distance to branching distance, in
future work.

\begin{proposition}[\cf~{\cite[Thm.~3.87]{aliprantis2007infinite}}]
  If trace distances $d_T^1$ and $d_T^2$ are Lipschitz equivalent,
  then the corresponding linear distances $d_L^1$ and $d_L^2$ are
  topologically equivalent.
\end{proposition}

\begin{proof}
  This follows immediately from Theorem~\ref{thm:dl-formula} and
  Theorem~3.87 in~\cite{aliprantis2007infinite}.  Note that the
  theorem in~\cite{aliprantis2007infinite} actually is stronger; it is
  enough to assume $d_T^1$ and $d_T^2$ to be \emph{uniformly
    equivalent}.
\end{proof}

The next theorem shows that if a trace distance can be used for
measuring trace differences beyond the first symbol (which will be the
case except for some especially trivial trace distances), then the
corresponding linear and branching distances are topologically
inequivalent.  The proof is an easy adaption of the standard argument
for the fact that language inclusion does not imply simulation.

\begin{definition}
  A trace distance $d_T: \K^\omega\times \K^\omega$ is said to be
  \emph{one-step indiscriminate} if $\sigma_0= \tau_0$ implies $d_T(
  \sigma, \tau)= 0$ for all $\sigma, \tau\in \K^\omega$.
\end{definition}

\begin{proposition}
  \label{pr:wts-ineq}
  If $d_T$ is not one-step indiscriminate, then there exists a
  weighted transition system $A$ on which the corresponding distances
  $d_L$ and $d_B$ are topologically inequivalent.
\end{proposition}

\begin{proof}
  Let $\sigma, \tau\in \K^\omega$ such that $\sigma_0= \tau_0$, $d_T(
  \sigma, \tau)> 0$, and $d_T( \tau, \sigma)> 0$.  $A$ is depicted in
  Figure~\ref{fi:wts-ineq}.

  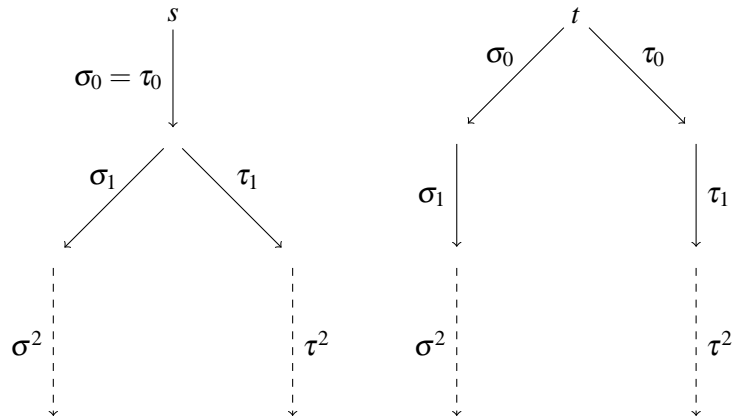
\begin{figure}[tp]
    \centering
    \begin{tikzpicture}[shorten >=1pt, auto, initial text=, scale=.2]
      \begin{scope}[outer sep=1pt,minimum size=5pt,inner sep=2pt, node
        distance=2.25cm]
        \node (v0) {$s$}; 
        \node (v1) [below of=v0, yshift=.6cm] {};
        \node (v2) [below right of=v1] {};
        \node (v3) [below left of=v1] {};
        \node (v4) [below of=v3] {};
        \node (v5) [below of=v2] {};
      \end{scope}
      \begin{scope}[->]
        \draw (v0) -- node[left]  {$\sigma_0 = \tau_0$} (v1);
        \draw (v1) -- node[above right, xshift=-.1cm] {$\tau_1$} (v2);
        \draw (v1) -- node[above left, xshift=.2cm]  {$\sigma_1$} (v3);
        \draw[dashed] (v3) -- node[left] {$\sigma^2$} (v4);
        \draw[dashed] (v2) -- node[right] {$\tau^2$}(v5);
      \end{scope}
    \end{tikzpicture}     
    \qquad
    \begin{tikzpicture}[->,shorten >=1pt, auto, node distance=.75cm,
      initial text=]
      \begin{scope}[outer sep=1pt,minimum
        size=5pt,inner sep=2pt, node distance=2.25cm] 
        \node (u0) {$t$};        
        \node (u1) [below left of=u0]{};
        \node (u2) [below right of=u0]{};
        \node (u3m) [below of=u1,yshift=.6cm] {}; 
        \node (u3) [below of=u3m] {};
        \node (u4m) [below of=u2,yshift=.6cm] {}; 
        \node (u4) [below of=u4m] {};    
      \end{scope}
      \begin{scope}
        \draw (u0) -- node[above,xshift=-.2cm]  {$\sigma_0$} (u1);
        \draw (u0) -- node[above,xshift=.2cm] {$\tau_0$} (u2);
        \draw (u2) -- node[right] {$\tau_1$} (u4m);
        \draw (u1) -- node[left] {$\sigma_1$} (u3m);
        \draw (u3m)[dashed] -- node[left] {$\sigma^2$} (u3);
        \draw (u4m)[dashed] -- node[right] {$\tau^2$} (u4);
      \end{scope}
    \end{tikzpicture}
    \caption{\label{fi:wts-ineq}%
      Weighted transition system for the proof of
      Proposition~\ref{pr:wts-ineq}}
  \end{figure}
  
  We have $\Tr( s)= \Tr(t)$, hence $d_L( s, t)= 0$.  On the other
  hand, $d_B( s, t)=\min\big( d_T( \sigma, \tau), d_T( \tau,
  \sigma)\big)> 0$.  As two equivalent hemimetrics have value $0$ at
  the same set of pairs of points, this finishes the proof.
\end{proof}

Also note that if $\sigma$ and $\tau$ are cyclic, the construction can
be adapted to yield a \emph{finite} WTS $A$.

\section{Recursively Defined Distances}
\label{sec:special}

The game definition of branching distance in
Definition~\ref{de:branch} gives a useful framework, but it is not
very operational.  In this section we show that if the given trace
distance has a recursive characterization, then the corresponding
branching distance can be obtained as the least fixed point of a
similar recursive formula.

We give the fixed-point theorem first and show in
Section~\ref{se:applications} below that the theorem covers all
examples of distances introduced earlier.

\begin{theorem}
  \label{th:F-recursion}
  Let $L$ be a complete lattice and $f: \K^\omega\times \K^\omega\to
  L$, $g: L\to[ 0, \infty]$, $F: \K\times \K\times L\to L$ such that
  $d_T= g\circ f$, $g$ is monotone, $F( x, y, \cdot): L\to L$ is
  monotone for all $x, y\in \K$, and
  \begin{equation}
    \label{eq:F-recursion}
    f( \sigma, \tau)= F\big( \sigma_0, \tau_0, f( \sigma^1,
    \tau^1)\big)
  \end{equation}
  for all $\sigma, \tau\in \K^\omega$.  Define $I: L^{ S\times S}\to
  L^{ S\times S}$ by
  \begin{equation*}
    I( h)( s, t)= \adjustlimits \sup_{ s\tto x s'} \inf_{ t\tto y t'}
    F\big( x, y, h( s', t')\big)
  \end{equation*}
  Then $I$ has a least fixed point $h^*: S\times S\to L$, and $d_B=
  g\circ h^*$.
\end{theorem}

Let us give some intuition about the theorem before we prove it.  Note
first the composition $d_T= g\circ f$, where $f$ maps pairs of traces
to the lattice $L$ which will act as \emph{memory} in the applications
below.  Equation~\eqref{eq:F-recursion} then expresses that $F$ acts
as a \emph{distance iterator function} which, within the lattice
domain, computes the trace distance by looking at the first elements
in the traces and then iterating over the rest of the trace.  Under
the premises of the theorem then, branching distance is the projection
by $g$ of the least fixed of a similar recursive function involving
$F$.

\begin{proof}
  It is not difficult to show that $I$ indeed has a least fixed point:
  The lattice $L^{ S\times S}$ with partial order defined point-wise
  by $h_1\le h_2$ iff $h_1( s, t)\le h_2( s, t)$ for all $s, t\in S$
  is complete, and $I$ is monotone because of the monotonicity
  condition on $F$, hence Tarski's Fixed-point Theorem can be applied.

  To show that $d_B= g\circ h^*$, we pull back $d_B$ along $g$: With
  the notation for the simulation game from Section~\ref{se:wsg},
  define
  \begin{equation*}
    f_B( s, t)= \adjustlimits \sup_{ \theta_1\in \Strat_1} \inf_{
      \theta_2\in \Strat_2} f\big( \tr( \pi_1( \theta_1, \theta_2)),
    \tr( \pi_2( \theta_1, \theta_2))\big)
  \end{equation*}
  We have $d_B= g\circ f_B$ by monotonicity of $g$, so it will suffice
  to show that $f_B= h^*$.

  Let us first prove that $f_B$ is a fixed point for $I$: Let $s, t\in
  S$, then
  \begin{align*}
    I( f_B)( s, t) &=%
    \adjustlimits \sup_{ s\tto x s'} \inf_{ t\tto y t'} F\big( x, y,
    f_B( s', t')\big)\\ &=%
    \adjustlimits \sup_{ s\tto x s'} \inf_{ t\tto y t'} F\big( x, y,
    \adjustlimits \sup_{ \theta_1'\in \Strat_1'} \inf_{ \theta_2'\in
      \Strat_2'} f( \tr( \pi_1( \theta_1', \theta_2')), \tr( \pi_2(
    \theta_1', \theta_2')))\big)\\ &=%
    \adjustlimits \sup_{ s\tto x s'} \inf_{ t\tto y t'} \adjustlimits
    \sup_{ \theta_1'\in \Strat_1'} \inf_{ \theta_2'\in \Strat_2'} F\big(
    x, y, f( \tr( \pi_1( \theta_1', \theta_2')), \tr( \pi_2(
    \theta_1', \theta_2')))\big)
  \end{align*}
  (the last step by the monotonicity assumption on $F$; note that in
  the second sup-inf pair, strategies from $s'$ and $t'$ are
  considered).  By the recursion formula~\eqref{eq:F-recursion} for
  $F$, we end up with
  \begin{equation*}
    I( f_B)( s, t) =%
    \adjustlimits \sup_{ s\tto x s'} \inf_{ t\tto y t'} \adjustlimits
    \sup_{ \theta_1'\in \Strat_1'} \inf_{ \theta_2'\in \Strat_2'} f(
    x\cdot \tr( \pi_1( \theta_1', \theta_2')), y\cdot \tr(
    \pi_2( \theta_1', \theta_2')))
  \end{equation*}
  Now because of independence of choices, we can rewrite this to
  \begin{equation*}
    I( f_B)( s, t) =%
    \sup_{ s\tto x s'} \sup_{ \theta_1'\in \Strat_1'} \inf_{ t\tto y
      t'} \inf_{ \theta_2'\in \Strat_2'} f(
    x\cdot \tr( \pi_1( \theta_1', \theta_2')), y\cdot \tr(
    \pi_2( \theta_1', \theta_2')))
  \end{equation*}
  and collapsing the sup-sup and inf-inf into one sup and inf,
  respectively, %
  conclude $I( f_B)= f_B$.

  To show that $f_B$ is the least fixed point for $I$, let $\bar h:
  S\times S\to L$ such that $I( \bar h)= \bar h$; we want to prove
  $f_B\le \bar h$.  Note first that for all $s, t\in S$ and all $s\tto
  x s'$, there is $t\tto y t'$ such that $F( x, y, \bar h( s', t'))\le
  I( \bar h)( s, t)$.  Now fix $s, t\in S$ and let $\theta_1\in
  \Strat_1$; we will be done once we can construct a Player-2 strategy
  $\theta_2\in \Strat_2$ for which $f( \tr( \pi_1( \theta_1,
  \theta_2)), \tr( \pi_2( \theta_1, \theta_2)))\le \bar h( s, t)$.

  We have to define $\theta_2$ for configurations $( \pi_1', \pi_2)\in
  \fPa( s)\times \fPa( t)$ in which $\pi_1'= \pi_1\cdot( s_j, x, s_{
    j+ 1})$ and $\len( \pi_1)= \len( \pi_2)$.  Write $\last( \pi_2)=
  t_j$; by the note above, we can choose a transition $t_j\tto y t_{
    j+ 1}$ for which $F( x, y, \bar h( s', t'))\le I( \bar h)( s, t)$,
  so we let $\theta_2( \pi_1', \pi_2)=( t_j, y, t_{ j+ 1})$.  The
  so-defined strategy has
  \begin{equation*}
    f( \tr( \pi_1( \theta_1, \theta_2)), \tr(
    \pi_2( \theta_1, \theta_2)))\le \adjustlimits \sup_{s\tto{x}s'}
    \inf_{t\tto{y}t'}F( x, y, \bar h( s', t'))= \bar h( s, t)
  \end{equation*}
\end{proof}

\subsection{Applications}
\label{se:applications}

We reconsider here the example trace distances from
Section~\ref{sec:trace_distances} and exhibit the corresponding linear
and branching distances.

\paragraph{Discrete trace distances.}

For the discrete trace distance on $\K^\omega$ given by $d_T( \sigma,
\tau)= 0$ if $\sigma= \tau$ and $d_T( \sigma, \tau)= \infty$
otherwise, we saw already in Section~\ref{se:discrete} that we recover
ordinary trace inclusion and simulation.  For linear distance, we can
also use Theorem~\ref{thm:dl-formula} to show that $d_L( s, t)= 0$ if
$\Tr( s)\subseteq \Tr( t)$ and $d_L( s, t)= \infty$ otherwise.

For the branching distance, we can now also apply
Theorem~\ref{th:F-recursion} with $L=[ 0, \infty]$, $g$ the identity
mapping, and $F( x, y, z)= z$ if $x= y$, $F( x, y, z)= \infty$
otherwise.  Then the branching distance is the least fixed point of
the equations $d_B( s, t)= \sup_{ s\tto x s'} \inf_{ t\tto x t'} d_B(
s', t')$, hence $d_B( s, t)= 0$ if $t$ simulates $s$ in the standard
sense~\cite{milner89}, and $d_B( s, t)= \infty$ otherwise.

For the refined discrete trace distance $d_T( \sigma, \tau)= 0$ if
$\sigma_j\sqsubseteq \tau_j$ for all $j$, $d_T( \sigma, \tau)= \infty$
otherwise, we analogously get $d_L( s, t)= 0$ if all $\sigma\in \Tr(
s)$ can be refined by a $\tau\in \Tr( t)$ (\ie~$\sigma_j\sqsubseteq
\tau_j$ for all $j$) and $d_L( s, t)= \infty$ otherwise.  Also, $d_B(
s, t)= 0$ if there is a relation $R\subseteq S\times S$ for which $(
s, t)\in R$ and whenever $( s', t')\in R$ and $s'\tto x s''$, then
also $t'\tto y t''$ with $x\sqsubseteq y$ and $( s'', t'')\in R$ (the
\emph{extended simulation} of~\cite{Thomsen87}), and $d_B( s, t)=
\infty$ otherwise.

\paragraph{Hamming distance.}

For Hamming distance induced by the metric $d( x, y)= 0$ if $x= y$ and
$d( x, y)= 1$ otherwise on $\K$, linear distance is given by $d_L( s,
t)\le k$ if and only if any trace $\sigma\in \Tr( s)$ can be matched
by a trace $\tau\in \Tr( t)$ with Hamming distance at most $k$, both
for the limit-average and the discounting interpretation.  The
branching distance associated with the discounting version is
precisely the (discounted) \emph{correctness distance}
of~\cite{conf/concur/CernyHR10}: $d_B( s, t)$ measures ``how often
[the system starting in $s_2$] can be forced to cheat'', \ie~to take a
transition different from the one the system starting in $s_1$ takes.

\paragraph{Labeled weighted transition systems.}

For the trace distances on labeled weighted transition systems, let us
for simplicity assume that $| \Sigma|= 1$, hence $\K= \Real$.  For the
point-wise trace distance $d_T( \sigma, \tau)= \sup_j| \sigma_j-
\tau_j|$ we can derive a recursive formula for the corresponding
branching distance by applying Theorem~\ref{th:F-recursion} with $L=[ 0,
\infty]$, $g$ the identity mapping, and $F( x, y, z)=\max\big(| x- y|,
z)$.  Then $d_B$ is the least fixed point to the equations
\begin{equation*}
  d_B( s, t)= \adjustlimits \sup_{ s\tto x s'} \inf_{ t\tto y
    t'} \max\big(| x- y|, d_B( s', t')\big)
\end{equation*}

For discounted accumulated trace distance $d_T( \sigma, \tau)= \sum_j
\lambda^j| \sigma_j- \tau_j|$, we can similarly let $F( x, y, z)=|
x- y|+ \lambda z$, then the corresponding branching distance is the
least fixed point to the equations
\begin{equation*}
  d_B( s, t)= \adjustlimits \sup_{ s\tto x s'} \inf_{ t\tto y
    t'} | x- y|+ \lambda d_B( s', t')
\end{equation*}
Note that these two branching distances are exactly the ones the
authors define in~\cite{journals/jlap/ThraneFL10}.

For the maximum-lead distance $d_T( \sigma, \tau)= \sup_j \bigl|
\sum_{i= 0}^j \sigma_i - \sum_{i= 0}^j \tau_i\bigr|$, we need to do
more work.  Intuitively, a recursive formulation needs to keep track
of the accumulated delay, hence needs (infinite) memory.  This can be
accomplished by letting $L=[ 0, \infty]^{[ -\infty, \infty]}$; the set
of functions from leads to distances.  We can then define
\begin{equation*}
  f( \sigma,
  \tau)( \delta)= \max\big(| \delta|, \sup_{ j= 0}^\infty| \delta+ \sum_{ i=
    0}^j \sigma_j- \sum_{ i= 0}^j \tau_j|\big)
\end{equation*}
and $g( h)= h( 0)$. Now with $F( x, y, h)( \delta)= \max\big(| \delta+
x- y|, h( \delta+ x- y)\big)$, we indeed have that $f( \sigma, \tau)=
F\big( \sigma_0, \tau_0, f( \sigma^1, \tau^1)\big)$, hence we can
apply Theorem~\ref{th:F-recursion} to conclude that $d_B( s, t)= h^*(
s, t)( 0)$, where $h^*$ is the least fixed point to the equations
\begin{equation*}
  h( s, t)( \delta)= \adjustlimits \sup_{ s\tto x s'} \inf_{ t\tto y
    t'} \max\big(| \delta+ x- y|, h( s', t')( \delta+ x- y)\big)
\end{equation*}
This is precisely the formulation of branching maximum-lead distance
given in~\cite{conf/formats/2005/HenzM05}.

\section{Conclusion and Future Work}

We have shown that simulation games with quantitative objectives
provide a general framework for studying linear and branching
distances for quantitative systems.  
Specifically, that our framework covers and unifies a number of
previously distinct approaches, and that certain common special cases
lead to useful recursive characterizations of branching distance.

Already we have seen that one very general property, topological
inequivalence of linear and branching distance, follows almost
immediately from the game characterization. Also this general approach
permits the conclusion that independent of the trace distance, the
branching distance provides an upper bound on the linear distance, a
property which is useful for applications such as analysis of
real-time systems, where linear distances are known to be
uncomputable~\cite{journals/jlap/ThraneFL10}.

It seems likely that by permitting a broader range of strategies, we
may encompass more advanced levels of system interaction and
observations, and hence capture quantitative extensions of other
well-known system relations such as 2-nested
simulation~\cite{DBLP:journals/iandc/GrooteV92,DBLP:conf/stacs/AcetoFI01}
or bisimulation~\cite{milner89}.  Thus, we expect our framework to be
of great use for reasoning about, and applying quantitative
verification.

The game perspective on linear and branching distances also suggests
that several interesting results and properties of games with
quantitative objectives are transferable to our setting.  As an
example, one may consider computability and complexity results: For a
concrete setting such as finite weighted labeled automata, discounted
or limit average accumulating distances can be computed using
discounted and mean-payoff games, respectively.  Hence the complexity
of computing these branching distances is in
$\textsf{NP}\cap\textsf{coNP}$.
Similarly, results concerning strategy iteration or value iteration
for games with quantitative objectives may be transferred to the
distance setting.

\bibliographystyle{eptcs}
\bibliography{linearbranchingbib}

\end{document}